\documentclass[10pt]{article}
\usepackage[utf8]{inputenc}
\usepackage{style}
\usepackage{framed}
\usepackage[parfill]{parskip}
\usepackage{setspace}
\usepackage{hyperref}
\allowdisplaybreaks

\title{Almost-Polynomial Ratio Almost-ETH-Hardness of Approximating {\sc Densest $k$-Subgraph}}

\author{
Pasin Manurangsi\thanks{Email: \texttt{pasin@berkeley.edu}. This material is based upon work supported by the National Science Foundation under Grants No. CCF 1540685 and CCF 1655215} \vspace{-0.5em}\\
UC Berkeley
}

\begin{document}

\maketitle
\thispagestyle{empty}

\begin{abstract}
In the {\sc Densest $k$-Subgraph} ({\sc D$k$S}) problem, given an undirected graph $G$ and an integer $k$, the goal is to find a subgraph of $G$ on $k$ vertices that contains maximum number of edges. Even though Bhaskara~\etal's state-of-the-art algorithm for the problem achieves only $O(n^{1/4 + \varepsilon})$ approximation ratio, previous attempts at proving hardness of approximation, including those under average case assumptions, fail to achieve a polynomial ratio; the best ratios ruled out under any worst case assumption and any average case assumption are only any constant (Raghavendra and Steurer) and $2^{O(\log^{2/3} n)}$ (Alon~\etal) respectively.

In this work, we propose the \emph{Almost-Exponential Time Hypothesis (Almost-ETH)}, which postulates that no algorithm can solve 3SAT on $n$ variables in $O(2^{n^{1 - \gamma}})$ time for any constant $\gamma > 0$. Assuming Almost-ETH, we show that, for any constant $\varepsilon > 0$, there is no polynomial-time algorithm that approximates {\sc D$k$S} to within $2^{(\log n)^{1 - \varepsilon}}$ factor of the optimum. In addition, our result has \emph{perfect completeness}, meaning that, assuming Almost-ETH, no polynomial-time algorithm can even distinguish between the case in which $G$ contains a $k$-clique and the case in which every induced $k$-subgraph of $G$ has density at most $2^{-(\log n)^{1 - \varepsilon}}$.
\end{abstract}

\newpage
\setcounter{page}{1}

\section{Introduction}

%TODO(pasin): write this!

%Some references:
%\begin{itemize}
%\item Approximation algorithms for DkS (general case):~\cite{KP93, FS97, SW98, FL01, FPK01, AHI02, GL09, BCCFV10}
%\item Approximation algorithms for DkS (almost perfect completeness case):~\cite{ST05, Barman15, MM15}
%\item Other approximation algorithms for DkS: Dense graph case~\cite{AKK95}
%\item Hardness of approximation of DkS:~\cite{Feige02, Kho06, RS10, AAMMW11, BKRW17}
%\item LP/SDP Lower bound on DkS:~\cite{BCVGZ12,M-thesis,CMMV17}
%\item Hardness of approximation based on DkS: 
%\begin{itemize}
%\item {\sc $k$-way Hypergraph Cut}, {\sc $k$-way Submodular Multiway}~\cite{CL15}
%\item {\sc Amphibious Influence Maximization}~\cite{CLLR15}
%\end{itemize}
%\item Birthday Repetition/Subexponential Size Reduction Papers:~\cite{AIM14,BKW15,R15,BPR16,MR16,R16,BKRW17,DFS16} ~\cite{FK04}
%\item Nearly-Linear Size PCP:~\cite{BS08,Din07,MR08}
%\item Original PCP:~\cite{AS98,ALMSS98}
%\item ETH/Gap-ETH:~\cite{IPZ01,Din16,MR16}
%\end{itemize}

In the {\sc Densest $k$-Subgraph} ({\sc D$k$S}) problem, we are given an undirected graph $G$ on $n$ vertices and a positive integer $k \leqs n$. The goal is to find a set $S$ of $k$ vertices such that the induced subgraph on $S$ has maximum number of edges. Since the size of $S$ is fixed, the problem can be equivalently stated as finding a $k$-subgraph (i.e. subgraph on $k$ vertices) with maximum density where density\footnote{It is worth noting that sometimes density is defined as $|E(S)|/|S|$. For {\sc D$k$S}, both definitions of density result in the same objective since $|S| = k$ is fixed. However, our notion is more convenient to deal with as it always lies in $[0, 1]$.} of the subgraph induced on $S$ is $|E(S)|/\binom{|S|}{2}$ and $E(S)$ denotes the set of all edges among the vertices in $S$.

{\sc Densest $k$-Subgraph}, a natural generalization of {\sc $k$-Clique}~\cite{Karp72}, was first formulated and studied by Kortsarz and Peleg~\cite{KP93} in the early 90s. Since then, it has been the subject of intense study in the context of approximation algorithm and hardness of approximation~\cite{FS97,SW98,FL01,FPK01,AHI02,Feige02,Kho06,GL09,RS10,BCCFV10,AAMMW11,BCVGZ12,Barman15,BKRW17}. Despite this, its approximability still remains wide open and is considered by some to be an important open question in approximation algorithms~\cite{BCCFV10,BCVGZ12,BKRW17}.

On the approximation algorithm front, Kortsarz and Peleg~\cite{KP93}, in the same work that introduced the problem, gave a polynomial-time $\tilde O(n^{0.3885})$-approximation algorithm for {\sc D$k$S}. Feige, Kortsarz and Peleg~\cite{FPK01} later provided an $O(n^{1/3 - \delta})$-approximation for the problem for some constant $\delta \approx 1/60$. This approximation ratio was the best known for almost a decade
\footnote{Around the same time as Bhaskara \etal's work~\cite{BCCFV10}, Goldstein and Langberg~\cite{GL09} presented an algorithm with approximation ratio $O(n^{0.3159})$, which is slightly better than~\cite{FPK01} but is worse than~\cite{BCCFV10}.} 
until Bhaskara~\etal~\cite{BCCFV10} invented a log-density based approach which yielded an $O(n^{1/4 + \varepsilon})$-approximation for any constant $\varepsilon > 0$. This remains the state-of-the-art approximation algorithm for {\sc D$k$S}.

While the above algorithms demonstrate the main progresses of approximations of {\sc D$k$S} in general case over the years, many special cases have also been studied. Most relevant to our work is the case where the optimal $k$-subgraph has high density, in which better approximations are known~\cite{FS97,ST05,MM15,Barman15}. The first and most representative algorithm of this kind is that of Feige and Seltser~\cite{FS97}, which provides the following guarantee: when the input graph contains a $k$-clique, the algorithm can find an $(1 - \varepsilon)$-dense $k$-subgraph in $n^{O(\log n / \varepsilon)}$ time. We will refer to this problem of finding densest $k$-subgraph when the input graph is promised to have a $k$-clique \emph{{\sc Densest $k$-Subgraph} with perfect completeness}.%In other words, the Feige-Seltser algorithm is a quasi-polynomial time approximation scheme for {\sc D$k$S} with perfect completeness.

%For instance, when $k$ is large relative to $n$, a greedy algorithm is shown to provide an $O(n/k)$-approximation to the problem~\cite{AHI02}; in this case, more complicated techniques have also been applied to improve the constant in the big-O notation~\cite{SW98, FL01}. Polynomial-time approximation schemes (PTAS) are also known for some classes of input graphs such as dense graphs~\cite{AKK95}, interval graphs~\cite{N16} and unit disk graphs~\cite{CFL11}. Lastly and most relevant to our work, when the optimal $k$-subgraph has high density, better approximations are known~\cite{FS97,ST05,MM15,Barman15}. The first and most representative algorithm of this kind is that of Feige and Seltser~\cite{FS97}, which provides the following guarantee: when the input graph contains a $k$-clique, the algorithm can find a $(1 - \varepsilon)$-dense $k$-subgraph in $n^{O(\log n / \varepsilon)}$ time. We call such problem of finding densest $k$-subgraph when the input graph is promised to have a $k$-clique \emph{{\sc Densest $k$-Subgraph} with perfect completeness}. In other words, the Feige-Seltser algorithm is a quasi-polynomial time approximation scheme (QPTAS) for {\sc D$k$S} with perfect completeness.

Although many algorithms have been devised for {\sc D$k$S}, relatively little is known regarding its hardness of approximation. While it is commonly believed that the problem is hard to approximate to within some polynomial ratio~\cite{AAMMW11,BCVGZ12}, not even a constant factor NP-hardness of approximation is known. To circumvent this, Feige~\cite{Feige02} came up with a hypothesis that a random 3SAT formula is hard to refute in polynomial time and proved that, assuming this hypothesis, {\sc D$k$S} is hard to approximate to within some constant factor. 

Alon~\etal~\cite{AAMMW11} later used a similar conjecture regarding random $k$-AND to rule out polynomial-time algorithms for {\sc D$k$S} with any constant approximation ratio. Moreover, they proved hardnesses of approximation of {\sc D$k$S} under the following \emph{Planted Clique Hypothesis}~\cite{Jer92,Kuc95}: there is no polynomial-time algorithm that can distinguish between a typical Erdős–Rényi random graph $\cG(n, 1/2)$ and one in which a clique of size polynomial in $n$ (e.g. $n^{1/3}$) is planted. Assuming this hypothesis, Alon \etals proved that no polynomial-time algorithm approximates {\sc D$k$S} to within any constant factor. They also showed that, when the hypothesis is strengthened to rule out not only polynomial-time but also super-polynomial time algorithms for the Planted Clique problem, their inapproximability guarantee for {\sc D$k$S} can be improved. In particular, if no $n^{O(\sqrt{\log n})}$-time algorithm solves the Planted Clique problem, then $2^{O(\log^{2/3} n)}$-approximation for {\sc D$k$S} cannot be achieved in polynomial time. 

There are also several inapproximability results of {\sc D$k$S} based on worst-case assumptions. Khot~\cite{Kho06} showed, assuming NP $\not\subseteq$ BPTIME($2^{n^\varepsilon}$) for some constant $\varepsilon > 0$, that no polynomial-time algorithm can approximate {\sc D$k$S} to within $(1 + \delta)$ factor where $\delta > 0$ is a constant depending only on $\varepsilon$; the proof is based on a construction of a ``quasi-random'' PCP, which is then used in place of a random 3SAT in a reduction similar to that from~\cite{Feige02}. 

While no inapproximability of {\sc D$k$S} is known under the Unique Games Conjecture, Raghavendra and Steurer~\cite{RS10} showed that a strengthened version of it, in which the constraint graph is required to satisfy a ``small-set expansion'' property, implies that {\sc D$k$S} is hard to approximate to within any constant ratio.

Braverman \etal~\cite{BKRW17} gave a $2^{\tilde\Theta(\sqrt{m})}$-time reduction from 3SAT with $m$ clauses to $(1 + \varepsilon)$-approximation of {\sc Densest $k$-Subgraph} with perfect completeness for some constant $\varepsilon > 0$. They originally used this reduction to rule out $n^{\tilde O(\log n)}$-time algorithms under the Exponential Time Hypothesis (ETH)~\cite{IP01}---which has recently been refuted~\cite{Ope26ETH}. Their reduction still rules out $O(n^{(\log n)^{1 - \delta}})$-time $(1 + \varepsilon)$-approximation algorithms under our weaker hypothesis (Almost-ETH, stated shortly) for any constant $\delta > 0$, nearly matching the previously mentioned Feige-Seltser algorithm~\cite{FS97}.

Since none of these inapproximability results achieve a polynomial ratio, there have been efforts to prove better lower bounds for more restricted classes of algorithms. For example, Bhaskara \etal~\cite{BCVGZ12} provided polynomial ratio lower bounds against strong SDP relaxations of {\sc D$k$S}. Specifically, for the Sum-of-Squares hierarchy, they showed integrality gaps of $n^{2/53 - \varepsilon}$ and $n^{\varepsilon}$ against $n^{\Omega(\varepsilon)}$ and $n^{1 - O(\varepsilon)}$ levels of the hierarchy respectively. (See also~\cite{M-thesis, CMMV17} in which $2/53$ in the exponent was improved to $1/14$.) Unfortunately, it is unlikely that these lower bounds can be translated to inapproximability results and the question of whether any polynomial-time algorithm can achieve subpolynomial approximation ratio for {\sc D$k$S} remains an intriguing open question.

\subsection{Our Results}

%In this work, we rule out, under ETH or Gap-ETH, polynomial-time approximation algorithms for {\sc D$k$S} (even with perfect completeness) with slightly subpolynomial ratio. In particular, assuming ETH, there is no polynomial-time $n^{1/(\loglog n)^c}$-approximation algorithm for the problem:

In this work, we propose the \emph{Almost-Exponential Time Hypothesis (Almost-ETH)}, which postulates that no algorithm can solve 3SAT on $n$ variables in $O(2^{n^{1 - \gamma}})$ time for any constant $\gamma > 0$ (see Hypothesis~\ref{hyp:eth}),\footnote{An earlier version of this paper assumed the Exponential Time Hypothesis (ETH)~\cite{IP01} and the Gap Exponential Time Hypothesis (Gap-ETH)~\cite{Din16,MR16} (obtaining inapproximability ratios of $n^{1/(\loglog n)^c}$ and $n^{o(1)}$, respectively), both of which are now known to be false~\cite{Ope26ETH}.} and rule out polynomial-time approximation algorithms for {\sc D$k$S} (even with perfect completeness) with almost-polynomial ratio $2^{(\log n)^{1 - \varepsilon}}$ for any constant $\varepsilon > 0$:
\begin{theorem} \label{thm:main-eth}
For every constant $\varepsilon \in (0, 1)$, assuming Almost-ETH, no polynomial-time algorithm can, given a graph $G$ on $n$ vertices and a positive integer $k \leqs n$, distinguish between the following two cases:
\begin{itemize}
\item There exist $k$ vertices of $G$ that induce a $k$-clique.
\item Every $k$-subgraph of $G$ has density at most $2^{-(\log n)^{1 - \varepsilon}}$.
\end{itemize}
\end{theorem}

We remark that, for {\sc D$k$S} with perfect completeness, the aforementioned Feige-Seltser algorithm achieves\footnote{This guarantee was not stated explicitly in~\cite{FS97} but it can be easily achieved by changing the degree threshold in their algorithm {\bf DenseSubgraph} from $(1 - \varepsilon)n$ to $n^\varepsilon$.} an $n^{\varepsilon}$-approximation in time $n^{O(1/\varepsilon)}$ for every $\varepsilon > 0$~\cite{FS97}. Hence, the ratio in Theorem~\ref{thm:main-eth} cannot be improved to some fixed polynomial.

{\bf Comparison to Previous Results.} In terms of the inapproximability ratio, the ratio ruled out in this work is almost polynomial and provides a vast improvement over previous results. Prior to our result, the best known ratio ruled out under any worst case assumption is only any constant~\cite{RS10} and the best ratio ruled out under any average case assumption is only $2^{O(\log^{2/3} n)}$~\cite{AAMMW11}. In addition, our result also has perfect completeness, which was only achieved in~\cite{BKRW17} and in~\cite{AAMMW11} under the Planted Clique Hypothesis but not in~\cite{Kho06,Feige02,RS10}.

Regarding the assumptions our result is based upon, the average case assumptions used in~\cite{Feige02,AAMMW11} are incomparable to ours. The assumption NP $\not\subseteq$ BPTIME($2^{n^\varepsilon}$) used in~\cite{Kho06} is also incomparable to ours since, while not stated explicitly, Almost-ETH by default focuses only on deterministic algorithms and our reduction is also deterministic. However, ignoring randomization, Almost-ETH is stronger as it rules out $O(2^{n^{1 - \gamma}})$-time algorithms for every $\gamma > 0$ rather than $O(2^{n^\varepsilon})$-time algorithms for \emph{some} $\varepsilon > 0$. The strengthened Unique Games Conjecture used in~\cite{RS10} is again incomparable to ours as it is a statement that a specific problem is NP-hard. 
Finally, although Braverman \etal's reduction~\cite{BKRW17} can also be used under Almost-ETH, its relation to our result is more subtle. Specifically, their reduction time is only $2^{\tilde\Theta(\sqrt{m})}$ where $m$ is the number of clauses, meaning that they only need to assume that 3SAT $\notin$ DTIME($2^{\tilde\Theta(\sqrt{m})}$) to rule out a constant ratio polynomial-time approximation for {\sc D$k$S}. However, as we will see in Theorem~\ref{thm:main-red}, even to achieve a constant gap, our reduction time is $2^{\tilde\Omega(m^{3/4})}$. Hence, if 3SAT somehow ends up in DTIME($2^{\tilde\Theta(m^{3/4})}$) but outside of DTIME($2^{\tilde\Theta(\sqrt{m})}$), their result will still hold whereas ours will not even imply constant ratio inapproximability for {\sc D$k$S}.

{\bf Implications of Our Results.} One of the reasons that {\sc D$k$S} has received significant attention in the approximation algorithm community is due to its connections to many other problems. Most relevant to our work are the problems to which there are reductions from {\sc D$k$S} that preserve approximation ratios to within some polynomial\footnote{These are problems whose $O(\rho)$-approximation gives an $O(\rho^c)$-approximation for {\sc D$k$S} for some constant $c$.}. These problems include {\sc Densest At-Most-$k$-Subgraph}~\cite{AC09}, {\sc Smallest $m$-Edge Subgraph}~\cite{CDK12}, {\sc Steiner $k$-Forest}~\cite{HJ06} and {\sc Quadratic Knapsack}~\cite{Pis07}. For brevity, we do not define these problems here. We refer interested readers to cited sources for their definitions and reductions from {\sc D$k$S} to respective problems. We also note that this list is by no means exhaustive and there are indeed numerous other problems with similar known connections to {\sc D$k$S} (see e.g.~\cite{HJL06,KS07,KMNT08,CHK,HIM11,LNV14,CLLR15,CL15,CZ15,SFL15,TV15,CDKKR16,CMVZ16,Lee16}). 
Our result also implies hardness of approximation results with similar ratio to {\sc D$k$S} for such problems:

\begin{corollary} \label{cor:impl-eth}
For every constant $\varepsilon > 0$, assuming Almost-ETH, there is no polynomial-time $2^{(\log n)^{1 - \varepsilon}}$-approximation algorithm for {\sc Densest At-Most-$k$-Subgraph}, {\sc Smallest $m$-Edge Subgraph}, {\sc Steiner $k$-Forest}, {\sc Quadratic Knapsack}.
\end{corollary}

\section{Preliminaries and Notations} \label{sec:prelim}

%Before we proceed to describe the reduction and prove our main theorems, we first define additional notations and provide some preliminaries.
We use $\exp(x)$ and $\log(x)$ to denote $e^x$ and $\log_2(x)$ respectively. $\polylog n$ is used as a shorthand for $O(\log^c n)$ for some constant $c$. For any set $S$, $\scrP(S) := \{T \mid T \subseteq S\}$ denotes the power set of $S$. For any non-negative integer $t \leqs |S|$, we use $\binom{S}{t} := \{T \in \scrP(S) \mid |T| = t\}$ to denote the collection of all subsets of $S$ of size $t$.

Throughout this work, we only concern with simple unweighted undirected graphs. Recall that the density of a graph $G = (V, E)$ on $N \geqs 2$ vertices is $|E|/\binom{N}{2}$. We say that a graph is \emph{$\alpha$-dense} if its density is $\alpha$. Moreover, for every $t \in \mathbb{N}$, we view each element of $V^t$ as a $t$-size ordered multiset of $V$. $(L, R) \in V^t \times V^t$ is said to be a \emph{labelled copy of a $t$-biclique} (or $K_{t, t}$) in $G$ if, for every $u \in L$ and $v \in R$, $u \ne v$ and $(u, v) \in E$. The number of labelled copies of $K_{t, t}$ in $G$ is the number of all such $(L, R)$'s.

\subsection{Almost-Exponential Time Hypothesis}

Impagliazzo and Paturi~\cite{IP01} proposed the Exponential Time Hypothesis (ETH), which conjectured that 3SAT on $n$ variables cannot be solved in $2^{o(n)}$ time. Recently, however, OpenAI~\cite{Ope26ETH} refuted ETH (as well as its strengthened gap version, Gap-ETH~\cite{Din16,MR16}) by giving a deterministic $2^{O(n \loglog n / \sqrt{\log n})}$-time algorithm for 3SAT. Since this running time is still of the form $2^{n^{1 - o(1)}}$, we propose the following weaker hypothesis:

\begin{hypothesis}[Almost-Exponential Time Hypothesis (Almost-ETH)] \label{hyp:eth}
For every constant $\gamma > 0$, no algorithm can decide whether any 3SAT formula on $n$ variables is satisfiable in $O(2^{n^{1 - \gamma}})$ time.
\end{hypothesis}

Similar to ETH~\cite{IPZ01}, Almost-ETH can be equivalently stated in terms of the number of clauses $m$ instead of the number of variables $n$:

\begin{lemma} \label{lem:almost-eth-clauses}
Assuming Almost-ETH, for every constant $\gamma > 0$, no algorithm can decide whether any 3SAT formula with $m$ clauses is satisfiable in $O(2^{m^{1 - \gamma}})$ time.
\end{lemma}

Since the proof of Lemma~\ref{lem:almost-eth-clauses} follows via a standard application of the Sparsification Lemma~\cite{IPZ01}, we defer it to Appendix~\ref{app:almost-eth-clauses}.

\subsection{Nearly-Linear Size PCPs and Subexponential Time Reductions}

The celebrated PCP Theorem~\cite{AS98,ALMSS98}, which lies at the heart of virtually all known NP-hardness of approximation results, can be viewed as a polynomial-time reduction from 3SAT to a gap version of 3SAT, as stated below. While this perspective is a rather narrow viewpoint of the theorem that leaves out the fascinating relations between parameters of PCPs, it will be the most convenient for our purpose.
\begin{theorem}[PCP Theorem~\cite{AS98,ALMSS98}]
For some constant $\varepsilon > 0$, there exists a polynomial-time reduction that takes a 3SAT formula $\varphi$ and produces a 3SAT formula $\phi$ such that
\begin{itemize}
\item \emph{(Completeness)} if $\varphi$ is satisfiable, then $\phi$ is satisfiable, and,
\item \emph{(Soundness)} if $\varphi$ is unsatisfiable, then $\val(\phi) \leqs 1 - \varepsilon$.
\end{itemize}
Here $\val(\phi)$ denotes the maximum fraction of clauses of $\phi$ satisfied by any assignment.
\end{theorem}

%In other words, the PCP Theorem states that the gap version of 3SAT is also NP-hard, i.e., it is NP-hard to approximate 3SAT to some constant ratio. 
Following the first proofs of the PCP Theorem, considerable efforts have been made to improve the trade-offs between the parameters in the theorem. One such direction is to try to reduce the size of the PCP, which, in the above formulation, translates to reducing the size of $\phi$ relative to $\varphi$. On this front, it is known that the size of $\phi$ can be made nearly-linear in the size of $\varphi$~\cite{Din07,MR08,BS08}. For our purpose, we will use Dinur's PCP Theorem~\cite{Din07}, which has a blow-up of only polylogarithmic in the size of $\phi$:

\begin{theorem}[Dinur's PCP Theorem~\cite{Din07}] \label{thm:dinur-pcp}
For some constant $\varepsilon, d > 0$, there exists a polynomial-time reduction that takes a 3SAT formula $\varphi$ with $m$ clauses and produces another 3SAT formula $\phi$ with $m' = O(m \polylog m)$ clauses such that
\begin{itemize}
\item \emph{(Completeness)} if $\varphi$ is satisfiable, then $\phi$ is satisfiable, and,
\item \emph{(Soundness)} if $\varphi$ is unsatisfiable, then $\val(\phi) \leqs 1 - \varepsilon$, and,
\item \emph{(Bounded Degree)} each variable of $\phi$ appears in $\leqs d$ clauses.
\end{itemize}
\end{theorem}

Note that Dinur's PCP, combined with Almost-ETH (via Lemma~\ref{lem:almost-eth-clauses}), implies that for every constant $\gamma > 0$, no algorithm can solve the gap version of 3SAT (even with bounded variable degree) in $O(2^{m^{1 - \gamma}})$ time.

Under Almost-ETH, nearly-linear size PCPs allow us to start with an instance $\phi$ of the gap version of 3SAT and reduce, in subexponential time, to another problem. As long as the time spent in the reduction is $O(2^{m^{1 - \gamma}})$ for some constant $\gamma > 0$, we arrive at a lower bound for the problem. Arguably, Aaronson \etal~\cite{AIM14} popularized this method, under the name \emph{birthday repetition}, by using such a reduction of size $2^{\tilde\Omega(\sqrt{m})}$ to prove hardness for free games and dense CSPs. %Note that this would not have worked with the original PCP Theorem as the size of $\phi$ can be larger than $m^2$, meaning that the reduction itself already takes exponential time. 
Without going into any detail now, let us mention that the name birthday repetition comes from the use of the birthday paradox in their proof and, since its publication, their work has inspired many inapproximability results~\cite{BKW15,R15,BPR16,MR16,R16,DFS16,BKRW17}. Our result too is inspired by this line of work and, as we will see soon, part of our proof also contains a birthday-type paradox. % but we would like to stress that significant new ideas are needed to make our proof work.

\section{The Reduction and Proof of the Main Theorem} \label{sec:proof}

The reduction from the gap version of 3SAT to {\sc D$k$S} is simple. Given a 3SAT formula $\phi$ on $n$ variables $x_1, \dots, x_n$ and an integer $1 \leqs \ell \leqs n$, we construct a graph\footnote{For interested readers, we note that our graph is not the same as the FGLSS graph~\cite{FGLSS} of the PCP in which the verifier reads $\ell$ random variables and accepts if no clause is violated; while this graph has the same vertex set as ours, the edges are different since we check that no clause between the two vertices is violated, which is not checked in the FGLSS graph. It is possible to modify our proof to make it work for this FGLSS graph. However, the soundness guarantee for the FGLSS graph is worse.} %However, the soundness there would become $2^{-\delta \ell^6 / n^5}$, which is worse than $2^{-\delta \ell^4 / n^3}$ we have in Theorem~\ref{thm:main-red} for our graph.} 
$G_{\phi, \ell} = (V_{\phi, \ell}, E_{\phi, \ell})$ as follows:
\begin{itemize}
\item Its vertex set $V_{\phi, \ell}$ contains all partial assignments to $\ell$ variables, i.e., each vertex is $\{(x_{i_1}, b_{i_1}), \dots, (x_{i_\ell}, b_{i_\ell})\}$ where $x_{i_1}, \dots, x_{i_\ell}$ are $\ell$ distinct variables and $b_{i_1}, \dots, b_{i_\ell} \in \{0, 1\}$ are the bits assigned to them.
\item We connect two vertices $\{(x_{i_1}, b_{i_1}), \dots, (x_{i_\ell}, b_{i_\ell})\}$ and \\ $\{(x_{i'_1}, b_{i'_1}), \dots, (x_{i'_\ell}, b_{i'_\ell})\}$ by an edge iff the two partial assignments are consistent (i.e. no variable is assigned 0 in one vertex and 1 in another), and, every clause in $\phi$ all of whose variables are from $x_{i_1}, \dots, x_{i_\ell}, x_{i'_1}, \dots, x_{i'_\ell}$ is satisfied by the partial assignment induced by the two vertices.
\end{itemize}

%\begin{definition} \label{def:rep-graph}
%Given a 3SAT formula $\phi$ on $n$ variables $x_1, \dots, x_n$. For any positive integer $1 \leqs \ell \leqs n$, we define an \emph{$\ell$-variable repetition} graph $G_{\phi, \ell} = (V_{\phi, \ell}, E_{\phi, \ell})$ of $\phi$ as follows:
%\begin{itemize}
%\item The vertex set $V_{\phi, \ell}$ consists of all possible assignments to $\ell$ distinct variables, i.e., each vertex in $V_{\phi, \ell}$ is of the form $\{(x_{i_1}, b_{i_1}), \dots, (x_{i_\ell}, b_{i_\ell})\}$ where $x_{i_1}, \dots, x_{i_\ell}$ are distinct variables and $b_{i_1}, \dots, b_{i_\ell}$ are the bits assigned to the variables. Note that the number of vertices $|V_{\phi, \ell}|$ is $\binom{n}{\ell} 2^\ell$.
%\item Two vertices $\{(x_{i_1}, b_{i_1}), \dots, (x_{i_\ell}, b_{i_\ell})\}$ and $\{(x_{i'_1}, b_{i'_1}), \dots, (x_{i'_\ell}, b_{i'_\ell})\}$ have an edge between them if and only if (1) the assignments to the variables are consistent, and, (2) every clause whose all three variables are from $x_{i_1}, \dots, x_{i_\ell}, x_{i'_1}, \dots, x_{i'_\ell}$ is satisfied by the assignment induced by the two vertices.
%\end{itemize}
%\end{definition}

Clearly, if $\val(\phi) = 1$, the $\binom{n}{\ell}$ vertices corresponding to a satisfying assignment  induce a clique. Our main technical contribution is proving that, when $\val(\phi) \leqs 1 - \varepsilon$, every $\binom{n}{\ell}$-subgraph is sparse:

\begin{theorem} \label{thm:main-red}
For any $d, \varepsilon > 0$, there exists $\delta > 0$ such that, for any 3SAT formula $\phi$ on $n$ variables such that $\val(\phi) \leqs 1 - \varepsilon$ and each variable appears in at most $d$ clauses and for any integer $\ell \in [n^{3/4} / \delta, n / 2]$, any $\binom{n}{\ell}$-subgraph of $G_{\phi, \ell}$ has density $\leqs 2^{-\delta \ell^4 / n^3}$.
\end{theorem}

We remark that there is nothing special about 3SAT; we can start with any boolean CSP and end up with a similar result, albeit the soundness deteriorates as the arity of the CSP grows. However, it is crucial that the variables are boolean; in fact, Braverman \etal~\cite{BKRW17} considered a graph similar to ours for 2CSPs but they were unable to achieve subconstant soundness since their variables were not boolean\footnote{Any satisfiable boolean 2CSP is solvable in polynomial time so one cannot start with a boolean 2CSP either.}. Specifically, there is a non-boolean 2CSP with low value which results in the graph having a biclique of size $\geqs \binom{n}{\ell}$ (Appendix~\ref{app:counter-bkrw}), i.e., one cannot get an inapproximability ratio more than two starting from a non-boolean CSP. 

Once we have Theorem~\ref{thm:main-red}, the inapproximability result of {\sc D$k$S} (Theorem~\ref{thm:main-eth}) can be easily proved by applying the theorem with an appropriate choice of $\ell$. We defer this proof to Subsection~\ref{subsec:dks-hardness}. For now, let us turn our attention to the proof of Theorem~\ref{thm:main-red}. To prove the theorem, we resort to the following lemma due to Alon~\cite{Alon02}, which states that every dense graph contains many labelled copies of bicliques:

\begin{lemma}[{\cite[Lemma 2.1]{Alon02}}\footnotemark] \label{lem:gen-kst}
Any $\alpha$-dense graph $G$ on $N \geqs 2$ vertices has at least $(\alpha/2)^{t^2}N^{2t}$ labelled copies of $K_{t, t}$ for all $t \in \mathbb{N}$.
\end{lemma}
\footnotetext{The lemma is stated slightly differently in~\cite{Alon02}. Namely, it was stated there that any graph $G$ with $\geqs \varepsilon N^2$ edges contains at least $(2\varepsilon)^{t^2}N^{2t}$ labelled copies of $K_{t, t}$. The formulation here follows from the fact that $\alpha$-dense graph on $N \geqs 2$ vertices contains at least $(\alpha/4)N^2$ edges.}

Equipped with Lemma~\ref{lem:gen-kst}, our proof strategy is to bound the number of labelled copies of $K_{t, t}$ in $G_{\phi, \ell}$ where $t$ is to be chosen later. To argue this, we will need some additional notations:
\begin{itemize}
\item First, let $A_{\phi} := \{(x_1, 0), (x_1, 1), \dots, (x_n, 0), (x_n, 1)\}$ be the set of all single-variable partial assignments. Observe that $V_{\phi, \ell} \subseteq \binom{A_{\phi}}{\ell}$, i.e., each $u \in V_{\phi, \ell}$ is a subset of $A_{\phi}$ of size $\ell$.
\item Let $\cA: (V_{\phi, \ell})^t \rightarrow \scrP(A_{\phi})$ be a ``flattening'' function that, on input $T \in (V_{\phi, \ell})^t$, outputs the set of all single-variable partial assignments that appear in at least one vertex in $T$. In other words, when each vertex $u$ is viewed as a subset of $A_{\phi}$, we can write $\cA(T)$ simply as $\bigcup_{u \in T} u$.
\item Let $\cK_{t, t} := \{(L, R) \in (V_{\phi, \ell})^t \times (V_{\phi, \ell})^t \mid \forall u \in L, \forall v \in R, u \ne v \wedge (u, v) \in E_{\phi, \ell}\}$ denote the set of all labelled copies of $K_{t, t}$ in $G_{\phi, \ell}$ and, for each $A, B \subseteq A_{\phi}$, let $\cK_{t, t}(A, B) := \{(L, R) \in \cK_{t, t} \mid \cA(L) = A, \cA(R) = B\}$ denote the set of all $(L, R) \in \cK_{t, t}$ with $\cA(L) = A$ and $\cA(R) = B$.
\end{itemize}
The number of labelled copies of $K_{t, t}$ in $G_{\phi, \ell}$ can be written as
\begin{align} \label{inq:ctt-to-cttalar}
|\cK_{t, t}| = \sum_{A, B \subseteq A_{\phi}} |\cK_{t, t}(A, B)|.
\end{align}
To bound $|\cK_{t, t}|$, we will prove the following bound on $|\cK_{t, t}(A, B)|$.

\begin{lemma} \label{lem:ctt-bound}
Let $\phi, n, \ell, d$ and $\varepsilon$ be as in Theorem~\ref{thm:main-red}. There exists $\lambda > 0$ depending only on $d$ and $\varepsilon$ such that, for any $t \in \mathbb{N}$ and any $A, B \subseteq A_{\phi}$, $|\cK_{t, t}(A, B)| \leqs \left(2^{- \lambda \ell^2 / n} \binom{n}{\ell}\right)^{2t}$.
\end{lemma}

Before we prove the above lemma, let us see how Lemma~\ref{lem:gen-kst} and Lemma~\ref{lem:ctt-bound} imply Theorem~\ref{thm:main-red}.

\begin{proof}[Proof of Theorem~\ref{thm:main-red}]
%We assume without loss of generality that $\lambda \leqs 1$. Pick $\delta = \lambda^2 / 12$ and $t = (4 / \lambda)(n^2 / \ell^2)$. From Lemma~\ref{lem:ctt-bound} and (\ref{inq:ctt-to-cttalar}), we can bound the number of copies of $K_{t, t}$ in $G_{\phi, \ell}$ by
%\begin{align*}
%|\cK_{t, t}|
%\leqs \sum_{A, B \subseteq A_{\phi}} \left(2^{-\lambda \ell^2/n}\binom{n}{\ell}\right)^{2t}
%= 2^{4n} \cdot \left(2^{-\lambda \ell^2/n}\binom{n}{\ell}\right)^{2t}
%\leqs (2^{-\lambda\ell^2/n})^{t} \cdot \binom{n}{\ell}^{2t}.
%\end{align*}
%where the last inequality comes from our choice of $t$; note that $t$ is chosen so that the $2^{4n}$ factor is consumed by $2^{-\lambda\ell^2/n}$ from Lemma~\ref{lem:ctt-bound}. Next, consider any $\binom{n}{\ell}$-subgraph of $G_{\phi, \ell}$. By the above bound, it contains at most $(2^{-\lambda\ell^2/n})^{t} \cdot \binom{n}{\ell}^{2t}$ copies of $K_{t, t}$. As a result, from Lemma~\ref{lem:gen-kst}, its density is at most $2 \cdot 2^{-\lambda\ell^2/(nt)}$.

Assume w.l.o.g. that $\lambda \leqs 1$. Pick $\delta = \lambda^2 / 8$ and $t = (4 / \lambda)(n^2 / \ell^2)$. From Lemma~\ref{lem:ctt-bound} and (\ref{inq:ctt-to-cttalar}), we have
\begin{align*}
|\cK_{t, t}| \leqs
2^{4n} \cdot \left(2^{-\lambda \ell^2/n}\binom{n}{\ell}\right)^{2t}
\leqs (2^{-\lambda\ell^2/n})^{t} \cdot \binom{n}{\ell}^{2t}
\end{align*}
where the second inequality comes from our choice of $t$; note that $t$ is chosen so that the $2^{4n}$ factor is consumed by $2^{-\lambda\ell^2/n}$ from Lemma~\ref{lem:ctt-bound}. Finally, consider any $\binom{n}{\ell}$-subgraph of $G_{\phi, \ell}$. By the above bound, it contains at most $(2^{-\lambda\ell^2/n})^{t} \cdot \binom{n}{\ell}^{2t}$ labelled copies of $K_{t, t}$. Thus, from Lemma~\ref{lem:gen-kst} and from $\ell \geqs n^{3/4}/\delta$, its density is at most $2 \cdot 2^{-\lambda\ell^2/(nt)} = 2 \cdot 2^{-2\delta \ell^4 / n^3} \leqs 2^{-\delta \ell^4 / n^3}$ as desired.
\end{proof}

We now move on to the proof of Lemma~\ref{lem:ctt-bound}.

\begin{proof}[Proof of Lemma~\ref{lem:ctt-bound}]
First, notice that if $(x, b)$ appears in $A$ and $(x, \neg b)$ appears in $B$ for some variable $x$ and bit $b$, then $\cK_{t, t}(A, B) = \emptyset$; this is because, for any $L$ with $\cA(L) = A$ and $R$ with $\cA(R) = B$, there exist $u \in L$ and $v \in R$ that contain $(x, b)$ and $(x, \neg b)$ respectively, meaning that there is no edge between $u$ and $v$ and, thus, $(L, R) \notin \cK_{t, t}(A, B)$. Hence, from now on, we can assume that, if $(x, b)$ appears in one of $A, B$, then the other does not contain $(x, \neg b)$. Observe that this implies that, for each variable $x$, its assignments can appear in $A$ and $B$ at most two times\footnote{This is where we use the fact that the variables are boolean. For non-boolean CSPs, each variable $x$ can appear more than two times in one of $A$ or $B$ alone, which can indeed be problematic (see Appendix~\ref{app:counter-bkrw}).} in total. This in turn implies that $|A| + |B| \leqs 2n$.

Let us now argue that $|\cK_{t, t}(A, B)|\leqs \binom{n}{\ell}^{2t}$; while this is not the bound we are looking for yet, it will serve as a basis for our argument later. For every $(L, R) \in \cK_{t, t}(A, B)$, observe that, since $\cA(L) = A$ and $\cA(R) = B$, we have $L \in \binom{A}{\ell}^t$ and $R \in \binom{B}{\ell}^t$. This implies that $\cK_{t, t}(A, B) \subseteq \binom{A}{\ell}^t \times \binom{B}{\ell}^t$. Hence, 
\begin{align} \label{inq:ctt-alar}
|\cK_{t, t}(A, B)| &\leqs \binom{|A|}{\ell}^t \binom{|B|}{\ell}^t.
\end{align}
Moreover, $\binom{|A|}{\ell} \binom{|B|}{\ell}$ can be further bounded as
\begin{align} \label{inq:alar-to-n}
\binom{|A|}{\ell} \binom{|B|}{\ell}
= \frac{1}{(\ell!)^2}\prod_{i=0}^{\ell-1}(|A| - i)(|B| - i)
\leqs \frac{1}{(\ell!)^2}\prod_{i=0}^{\ell-1}\left(\frac{|A| + |B|}{2} - i\right)^2\leqs \binom{n}{\ell}^2
\end{align}
where the inequalities come from the AM-GM Inequality and from $|A| + |B| \leqs 2n$ respectively.
Combining (\ref{inq:ctt-alar}) and (\ref{inq:alar-to-n}) indeed yields $|\cK_{t, t}(A, B)|\leqs \binom{n}{\ell}^{2t}$.

Inequality (\ref{inq:ctt-alar}) is very crude; we include all elements of $\binom{A}{\ell}$ and $\binom{B}{\ell}$ as candidates for vertices in $L$ and $R$ respectively. However, as we will see soon, only tiny fraction of elements of $\binom{A}{\ell}, \binom{B}{\ell}$ can actually appear in $L, R$ when $(L, R) \in \cK_{t, t}(A, B)$. To argue this, let us categorize the variables into three groups:
\begin{itemize}
\item $x$ is \emph{terrible} iff its assignments appear at most once in total in $A$ and $B$ (i.e. $|\{(x, 0), (x, 1)\} \cap A| + |\{(x, 0), (x, 1)\} \cap B| \leqs 1$).
\item $x$ is \emph{good} iff, for some $b \in \{0, 1\}$, $(x, b) \in A \cap B$. Note that this implies that $(x, \neg b) \notin A \cup B$.
\item $x$ is \emph{bad} iff either $\{(x, 0), (x, 1)\} \subseteq A$ or $\{(x, 0), (x, 1)\} \subseteq B$.% Note here also that no assignment of $x$ can appear in one of $A, B$.
\end{itemize}

The next and last step of the proof is where birthday-type paradoxes come in. Before we continue, let us briefly demonstrate the ideas behind this step by considering the following extreme cases:
\begin{itemize}
\item If all variables are terrible, then $|A| + |B| \leqs n$ and (\ref{inq:alar-to-n}) can be immediately tightened.
\item If all variables are bad, assume w.l.o.g. that, for at least half of variables $x$'s, $\{(x, 0), (x, 1)\} \subseteq A$. Consider a random element $u$ of $\binom{A}{\ell}$. Since $u$ is a set of random $\ell$ distinct elements of $A$, there will, in expectation, be $\Omega(\ell^2/n)$ variables $x$'s with $(x, 0), (x, 1) \in u$. However, the presence of such $x$'s means that $u$ is not a valid vertex. Moreover, it is not hard to turn this into the following probabilistic statement: with probability at most $2^{-\Omega(\ell^2 / n)}$, $u$ contains at most one of $(x, 0), (x, 1)$ for every variable $x$. In other words, only $2^{-\Omega(\ell^2 / n)}$ fraction of elements of $\binom{A}{\ell}$ are valid vertices, which yields the desired bound on $|\cK_{t, t}(A, B)|$.
\item If all variables are good, then $A = B$ is simply an assignment to all the variables. Since $\val(\phi) \leqs 1 - \varepsilon$, at least $\epsilon m$ clauses are unsatisfied by this assignment. As we will argue below, every element of $\binom{A}{\ell}$ that contains two variables from some unsatisfied clause cannot be in $L$ for any $(L, R) \in \cK_{t, t}(A, B)$. This means that there are $\Theta_{\varepsilon}(m) \geqs \Omega_{\varepsilon}(n)$ prohibited pairs of variables that cannot appear together. Again, similar to the previous case, it is not hard to argue that only $2^{-\Omega_{\varepsilon, d}(\ell^2/n)}$ fraction of elements of $\binom{A}{\ell}$ can be candidates for vertices of $L$.
\end{itemize}
%If a non-negligible fraction of the variables are terrible, then $|A| + |B|$ is going to be noticeably smaller than $n$, which means that the bound given above can be tightened immediately. On the other hand, if a non-negligible fraction of the variables are bad, then most of the elements of $\binom{A}{\ell}$ or $\binom{B}{\ell}$ will not be valid vertices because they contain two different assignments to some variable. Finally, if we assume the other extreme that all the variables are good, then $A = B$ is just an assignment. Since $\val(\phi) \leqs 1 - \varepsilon$, $\varepsilon$ fraction of the clauses are not satisfied by this assignment and hence we will see these clauses in almost all the elements of $\binom{A}{\ell} = \binom{B}{\ell}$. These are isolated vertices and cannot appear in any bicliques.

To turn this intuition into a bound on $|\cK_{t, t}(A, B)|$, we need the following inequality. Its proof is straightforward and is deferred to Subsection~\ref{subsec:birthday}.

\begin{proposition} \label{prop:birthday}
Let $U$ be any set and $P \subseteq \binom{U}{2}$ be any set of pairs of elements of $U$ such that each element of $U$ appears in at most $q$ pairs. For any positive integer $2 \leqs r \leqs |U|$, the probability that a random element of $\binom{U}{r}$ does not contain both elements of any pair in $P$ is at most $\exp\left(-\frac{|P|r^2}{4q|U|^2}\right)$.
\end{proposition}

We are now ready to formalize the above intuition and finish the proof of Lemma~\ref{lem:ctt-bound}. For the sake of convenience, denote the sets of good, bad and terrible variables by $X_g, X_b$ and $X_t$ respectively. Moreover, let $\beta := \varepsilon / (100d)$ and pick $\lambda = \min\{- \log(1 - \beta / 2), \beta / 64, \varepsilon / (384d)\}$. To refine the bound on the size of $\cK_{t, t}(A, B)$, consider the following three cases:
\begin{enumerate}[itemsep=6pt]
\item $|X_t| \geqs \beta n$. Since each $x \in X_t$ contributes at most one to $|A| + |B|$, $|A| + |B| \leqs (1 - \beta/2)(2n)$. Hence, we can improve (\ref{inq:alar-to-n}) to $\binom{|A|}{\ell}\binom{|B|}{\ell} \leqs \binom{(1 - \beta/2)n}{\ell}^2$. Thus, we have
\begin{align*}
|\cK_{t, t}(A, B)| 
\overset{(\ref{inq:ctt-alar})}{\leqs} \binom{|A|}{\ell}^t\binom{|B|}{\ell}^t  \leqs
\binom{(1 - \beta/2)n}{\ell}^{2t} \leqs \left((1 - \beta/2)^\ell\binom{n}{\ell}\right)^{2t} \leqs \left(2^{-\lambda \ell^2 / n}\binom{n}{\ell}\right)^{2t}
\end{align*}
where the last inequality comes from $\lambda \leqs - \log(1 - \beta / 2)$ and $\ell > \ell^2 / n$.
\item $|X_b| \geqs \beta n$. Since each $x \in X_b$ appears either in $A$ or $B$, one of $A$ and $B$ must contain assignments to at least $(\beta / 2) n$ variables in $X_b$. Assume w.l.o.g. that $A$ satisfies this property. Let $X_b^L$ be the set of all $x \in X_b$ whose assignments appear in $A$. Note that $|X_b^L| \geqs (\beta / 2) n$.

Observe that an element $u \in \binom{A}{\ell}$ is not a valid vertex if it contains both $(x, 0)$ and $(x, 1)$ for some $x \in X_b^L$. We invoke Proposition~\ref{prop:birthday} with $U = A$, $P = \{\{(x, 0), (x, 1)\} \mid x \in X_b^L\}, q = 1$ and $r = \ell$, which implies that a random element of $\binom{A}{\ell}$ does not contain any prohibited pairs in $P$ with probability at most $\exp\left(-\frac{|X_b^L|\ell^2}{4|A|^2}\right) \leqs \exp\left(-\frac{(\beta/2)n\ell^2}{4(2n)^2}\right)$, which is at most $2^{-2\lambda \ell^2/n}$ because $\lambda \leqs \beta /64$. In other words, at most $2^{-2\lambda \ell^2/n}$ fraction of elements of $\binom{A}{\ell}$ are valid vertices. This gives us the following bound:
\begin{align*}
|\cK_{t, t}(A, B)| \leqs 
\left(2^{-2\lambda \ell^2/n} \cdot \binom{|A|}{\ell}\right)^t \cdot \binom{|B|}{\ell}^t \overset{(\ref{inq:alar-to-n})}{\leqs} \left(2^{-\lambda \ell^2/n} \binom{n}{\ell}\right)^{2t}.
\end{align*}
\item $|X_t| < \beta n$ and  $|X_b| < \beta n$. In this case, $|X_g| > (1 - 2 \beta)n$. Let $S$ denote the set of clauses whose variables all lie in $X_g$. Since each variable appears in at most $d$ clauses, $|S| > m - (2 \beta n) d \geqs (1 - \varepsilon / 2) m$ where the second inequality comes from our choice of $\beta$ and from $m \geqs n / 3$.

Consider the partial assignment $f: X_g \rightarrow \{0, 1\}$ induced by $A$ and $B$, i.e., $f(x) = b$ iff $(x, b) \in A, B$. Since $\val(\phi) \leqs 1 - \varepsilon$, the number of clauses in $S$ satisfied by $f$ is at most $(1 - \varepsilon)m$. Hence, at least $\varepsilon m / 2$ clauses in $S$ are unsatisfied by $f$. Denote the set of such \iffalse unsatisfied\fi clauses by $S_{\text{UNSAT}}$.

Fix a clause $C \in S_{\text{UNSAT}}$ and let $x, y$ be two different variables in $C$. We claim that $x, y$ cannot appear together in any vertex of $L$ for any $(L, R) \in \cK_{t, t}(A, B)$. Suppose for the sake of contradiction that $(x, f(x))$ and $(y, f(y))$ both appear in $u \in L$ for some $(L, R) \in \cK_{t, t}(A, B)$. Let $z \in X_g$ be another variable\footnote{If $C$ contains two variables, let $z = x$. Note that we can assume w.l.o.g. that $C$ contains at least two variables.} in $C$. Since $(z, f(z)) \in B$, some vertex $v \in R$ contains $(z, f(z))$. Thus, there is no edge between $u$ and $v$ in $G_{\phi, \ell}$, which contradicts with $(L, R) \in \cK_{t, t}$.

%In other words, such pairs $(x, f(x)), (y, f(y))$ are prohibited pairs that cannot appear together in any vertex of $L$ for any $(L, R) \in \cK_{t, t}(A, B)$. Since $|S_{\text{UNSAT}}| \geqs \varepsilon m / 2 \geqs \varepsilon n / 6$, the number of such prohibited pairs is also at least $\varepsilon n / 6$. Moreover, since each variable appears in at most $d$ clauses, each $(x, f(x))$ appears in at most $2d$ prohibited pairs.

We can now appeal to Proposition~\ref{prop:birthday} with $U = A$, $q = 2d$, $r = \ell$ and $P$ be the prohibited pairs described above. This implies that with probability at most $\exp\left(-\frac{|P|\ell^2}{8d|A|^2}\right) \leqs \exp\left(-\frac{\varepsilon \ell^2}{192 d n}\right)$, a random element of $\binom{A}{\ell}$ contains no prohibited pair from $P$. In other words, at most $\exp\left(-\frac{\varepsilon \ell^2}{192 d n}\right)$ fraction of elements of $\binom{A}{\ell}$ can be candidates for each element of $L$ for $(L, R) \in \cK_{t, t}(A, B)$. This gives the following bound:
\begin{align*}
|\cK_{t, t}(A, B)| \leqs 
\left(\exp\left(-\frac{\varepsilon \ell^2}{192 d n}\right) \cdot \binom{|A|}{\ell}\right)^t \cdot \binom{|B|}{\ell}^t \overset{(\ref{inq:alar-to-n})}{\leqs} \left(2^{- \varepsilon \ell^2 / (384 dn)} \cdot \binom{n}{\ell}\right)^{2t}.
\end{align*}

Since we picked $\lambda \leqs \varepsilon / (384d)$, $|\cK_{t, t}(A, B)|$ is once again bounded above by $\left(2^{-\lambda \ell^2 / n}\binom{n}{\ell}\right)^{2t}$.
\end{enumerate}
In all three cases, we have $|\cK_{t, t}(A, B)| \leqs \left(2^{-\lambda \ell^2 / n}\binom{n}{\ell}\right)^{2t}$, completing the proof of Lemma~\ref{lem:ctt-bound}.
\end{proof}

\subsection{Proof of Proposition~\ref{prop:birthday}} \label{subsec:birthday}

\begin{proof}[Proof of Proposition~\ref{prop:birthday}]
We first construct $P' \subseteq P$ such that each element of $U$ appears in at most one pair in $P'$ as follows. Start out by marking every pair in $P$ as active and, as long as there are active pairs left, include one in $P'$ and mark every pair that shares an element of $U$ with this pair as inactive. Since each element of $U$ appears in at most $q$ pairs in $P$, we mark at most $2q$ pairs as inactive per each inclusion. This implies that $|P'| \geqs |P| / (2q)$.

Suppose that $P' = \{\{a_1, b_1\}, \dots, \{a_{|P'|}, b_{|P'|}\}\}$ where $a_1, b_1, \dots, a_{|P'|}, b_{|P'|}$ are distinct elements of $U$. Let $u$ be a random element of $\binom{U}{r}$. For each $i = 1, \dots, |P'|$, we have
\begin{align*}
\Pr[\{a_i, b_i\} \not\subseteq u] %&= 1 - \Pr[\{a_i, b_i\} \subseteq u]\\
&= 1 - \frac{\binom{|U| - 2}{r - 2}}{\binom{|U|}{r}} \\
&= 1 - \frac{r(r - 1)}{|U|(|U| - 1)} \\
(\text{Since } r - 1 \geqs r / 2 \text{ for all } r \geqs 2)&\leqs 1 - \frac{r^2}{2|U|^2} \\
(\text{Since } 1 - z \leqs \exp(-z) \text{ for all } z \in \mathbb{R}) &\leqs \exp\left(-\frac{r^2}{2|U|^2}\right).
\end{align*}

If $u$ does not contain both elements of any pairs in $P$, it does not contain both elements of any pairs in $P'$. The probability of the latter can be written as $$\Pr\left[\bigwedge_{i=1}^{|P'|} \{a_i, b_i\} \not\subseteq u\right] = \prod_{i=1}^{|P'|} \Pr\left[\{a_i, b_i\} \not\subseteq u \middle|\: \bigwedge_{j=1}^{i-1} \{a_j, b_j\} \not\subseteq u\right].$$ In addition, since $a_1, b_1, \dots, a_{|P'|}, b_{|P'|}$ are distinct, it is not hard to see that $\Pr\left[\{a_i, b_i\} \not\subseteq u \middle|\: \bigwedge_{j=1}^{i-1} \{a_j, b_j\} \not\subseteq u\right] \leqs \Pr[\{a_i, b_i\} \not\subseteq u]$. Hence, we have
\begin{align*}
\Pr\left[\bigwedge_{i=1}^{|P'|} \{a_i, b_i\} \not\subseteq u\right]
= \prod_{i=1}^{|P'|} \Pr[\{a_i, b_i\} \not\subseteq u]
\leqs \left(\exp\left(-\frac{r^2}{2|U|^2}\right)\right)^{|P'|}
= \exp\left(-\frac{|P'|r^2}{2|U|^2}\right)
\leqs \exp\left(-\frac{|P|r^2}{4q|U|^2}\right),
\end{align*}
completing the proof of Proposition~\ref{prop:birthday}.
\end{proof}

\subsection{Proof of Inapproximability Result of {\sc D$k$S}} \label{subsec:dks-hardness}

In this subsection, we prove Theorem~\ref{thm:main-eth} by combining Dinur's PCP Theorem and Theorem~\ref{thm:main-red} with $\ell = \lfloor n^{1 - \gamma} \rfloor$ for a small constant $\gamma > 0$, as stated below.

\begin{proof}[Proof of Theorem~\ref{thm:main-eth}]
Let $\varepsilon_0, d > 0$ be the constants from Dinur's PCP Theorem (Theorem~\ref{thm:dinur-pcp}), and let $\delta > 0$ be the constant from Theorem~\ref{thm:main-red} corresponding to $d$ and $\varepsilon_0$. Set $\gamma := \varepsilon / 5 \in (0, 1/5)$.

For any 3SAT formula $\varphi$ with $m$ clauses, use Dinur's PCP Theorem (Theorem~\ref{thm:dinur-pcp}) to produce a 3SAT formula $\phi$ with $m' = O(m \polylog m)$ clauses such that each variable appears in at most $d$ clauses. Let $n \leqs 3m'$ be the number of variables of $\phi$. We assume w.l.o.g. that $n$ is sufficiently large. Let $\ell = \lfloor n^{1 - \gamma} \rfloor$; since $\gamma < 1/4$, we have $n^{3/4} / \delta \leqs \ell \leqs n/2$ and $\ell \geqs n^{1 - \gamma} / 2$. Let us consider the graph $G_{\phi, \ell}$ with $k = \binom{n}{\ell}$, and let $N$ be the number of vertices of $G_{\phi, \ell}$. Observe that
\begin{align*}
N = 2^{\ell}\binom{n}{\ell} \leqs n^{2\ell} \leqs 2^{2 n^{1 - \gamma} \log n} = 2^{o(m^{1 - \gamma/2})},
\end{align*}
where the last equality follows from $n \leqs 3m' = O(m \polylog m)$. Note also that constructing $G_{\phi, \ell}$ takes $\poly(N, m') = 2^{o(m^{1 - \gamma/2})}$ time.

If $\varphi$ is satisfiable, $\phi$ is also satisfiable and it is obvious that $G_{\phi, \ell}$ contains an induced $k$-clique. Otherwise, if $\varphi$ is unsatisfiable, $\val(\phi) \leqs 1 - \varepsilon_0$. From Theorem~\ref{thm:main-red}, any $k$-subgraph of $G_{\phi, \ell}$ has density at most $2^{-\delta \ell^4 / n^3} \leqs 2^{-(\delta / 16) n^{1 - 4\gamma}}$. Meanwhile, since $\varepsilon = 5\gamma$, we have $(\log N)^{1-\varepsilon} \leqs (2n)^{1-5\gamma}$, which is at most $(\delta / 16) n^{1 - 4\gamma}$
for any sufficiently large $n$. Thus, the density of any $k$-subgraph of $G_{\phi, \ell}$ is at most $2^{-(\log N)^{1 - \varepsilon}}$. Hence, if there is a polynomial-time algorithm that can distinguish between the two cases in Theorem~\ref{thm:main-eth}, then there also exists an algorithm that solves 3SAT with $m$ clauses in $O(2^{m^{1 - \gamma/2}})$ time, contradicting Lemma~\ref{lem:almost-eth-clauses} (and thus Almost-ETH).
\end{proof}

\section{Conclusion and Open Questions} \label{sec:open}

In this work, we provide a subexponential time reduction from the gap version of 3SAT to {\sc D$k$S} and prove that it establishes an almost-polynomial ratio hardness of approximation of the latter under Almost-ETH. Even with our result, however, approximability of {\sc D$k$S} still remains wide open. Namely, it is still not known whether it is NP-hard to approximate {\sc D$k$S} to within some constant factor, and, no polynomial ratio hardness of approximation is yet known. 

Although our result appears to almost resolve the second question, it still seems out of reach with our current knowledge of hardness of approximation. In particular, to achieve a polynomial ratio hardness for {\sc D$k$S}, it is plausible that one has to prove a long-standing conjecture called \emph{the sliding scale conjecture (SSC)}~\cite{BGLR93}. In short, SSC essentially states that {\sc Label Cover}, a problem used as starting points of almost all NP-hardness of approximation results, is NP-hard to approximate to within some polynomial ratio. Note here that polynomial ratio hardness for {\sc Label Cover} is not even known under stronger assumptions such as Almost-ETH (or even the now-refuted ETH and Gap-ETH; see~\cite{Din16} for more detailed discussions on the topic).

%There is in fact an approximation preserving reduction from {\sc D$k$S} to {\sc Label Cover}~\cite{CHK} but it does not provide perfect completeness, which is required in SSC; this leaves a possibility that a polynomial ratio hardness of approximation of {\sc D$k$S} can be achieved without resolving SSC.

Apart from the approximability of {\sc D$k$S}, our result also prompts the following natural question: since previous techniques, such as Feige's Random 3SAT Hypothesis~\cite{Feige02}, Khot's Quasi-Random PCP~\cite{Kho06}, Unique Games with Small Set Expansion Conjecture~\cite{RS10} and the Planted Clique Hypothesis~\cite{Jer92,Kuc95}, that were successful in showing inapproximability of {\sc D$k$S} also gave rise to hardnesses of approximation of many problems that are not known to be APX-hard including {\sc Sparsest Cut}, {\sc Min Bisection}, {\sc Balanced Separator}, {\sc Minimum Linear Arrangement} and {\sc 2-Catalog Segmentation}~\cite{AMS07,Saket10,RST12}, is it possible to modify our construction to prove inapproximability for these problems as well? An evidence suggesting that this may be possible is the case of $\varepsilon$-approximate Nash Equilibrium with $\varepsilon$-optimal welfare, which was first proved to be hard under the Planted Clique Hypothesis by Hazan and Krauthgamer~\cite{HK11} before Braverman, Ko and Weinstein~\cite{BKW15} gave a $2^{\tilde\Theta(\sqrt{m})}$-time reduction from 3SAT showing that the problem is also hard under Almost-ETH.

%Apart from the approximability of {\sc D$k$S}, our results also prompt the following question: since previous techniques, such as Feige's Random 3SAT Hypothesis~\cite{Feige02}, Khot's Quasi-Random PCP~\cite{Kho06}, the Small Set Expansion Conjecture~\cite{RS10} and the Planted Clique Hypothesis~\cite{Jer92,Kuc95}, that were successful in showing inapproximability of {\sc D$k$S} also gave rise to hardnesses of approximation of many problems that are not known to be APX-hard such as {\sc Sparsest Cut} and {\sc Min Bisection}~\cite{AMS07,RST12}, is it possible to modify our construction to prove inapproximability for these problems as well? %An evidence suggesting that this may be possible is the case of $\varepsilon$-approximate Nash Equilibrium with $\varepsilon$-optimal welfare, which was first proved to be hard under the Planted Clique Hypothesis by Hazan and Krauthgamer~\cite{HK11} before Braverman~\etal~proved a similar hardness under ETH~\cite{BKW15}.

\subsection*{Acknowledgement}

I am truly grateful to Aviad Rubinstein, Prasad Raghavendra and Luca Trevisan for invaluable discussions throughout various stages of this work. Without their comments, suggestions and support, this work would not have been possible. Furthermore, I thank Daniel Reichman and Igor Shinkar for stimulating dicussions on a related problem which inspire part of the proof presented here. I also thank Igor for pointing me to~\cite{Alon02}. Finally, I thank anonymous reviewers for their useful comments on an earlier draft of this work.

\bibliographystyle{alpha}
\bibliography{main}

\appendix

\section{A Counterexample to Obtaining a Subconstant Soundness from Non-Boolean CSPs} \label{app:counter-bkrw}

Here we sketch an example due to Rubinstein~\cite{Rub16com} of a non-boolean 2CSP $\phi$ with low value for which the graph $G_{\phi, \ell}$ contains a large biclique. For a non-boolean 2CSP, we define the graph $G_{\phi, \ell}$ similar to that of a 3SAT formula except that now the vertices contains all $\{(x_{i_1}, \sigma_{i_1}), \dots, (x_{i_\ell}, \sigma_{i_\ell})\}$ for all distinct variables $x_{i_1}, \dots, x_{i_\ell}$ and all $\sigma_{i_1}, \dots, \sigma_{i_\ell} \in \Sigma$ where $\Sigma$ is the alphabet of the CSP.

%Below we describe an example of a non-boolean CSP $\phi$ with low value for which the graph $G_{\phi, \ell}$ contains a large biclique. This example is due to Aviad Rubinstein~\cite{Rub16com}. We note that, for a non-boolean CSP, we define the graph $G_{\phi, \ell}$ similar to that of a 3SAT formula. The only difference is that now the vertices contains all $\{(x_{i_1}, \sigma_{i_1}), \dots, (x_{i_\ell}, \sigma_{i_\ell})\}$ for all distinct variables $x_{i_1}, \dots, x_{i_\ell}$ and all $\sigma_{i_1}, \dots, \sigma_{i_\ell} \in \Sigma$ where $\Sigma$ is the alphabet of the CSP.

Consider any non-boolean 2CSP instance $\phi$ on variables $x_1, \dots, x_n$ such that there is no constraint between $X_1 := \{x_1, \dots, x_{n/2}\}$ and $X_2 = \{x_{n/2 + 1}, \dots, x_{n}\}$ and each variable appears in $\leqs d$ constraints. Let $L$ the set of all vertices $u$ such that every variable in $u$ belongs to $X_1$ and no constraint is contained in $u$. Define $R$ similarly for $X_2$. Clearly, $(L, R)$ forms a biclique and it is not hard to see that $|L|, |R| \geqs |\Sigma|^{\ell}\binom{n/2 - (d + 1)\ell}{\ell}$. Since $|\Sigma| \geqs 3$, this value is $\geqs \binom{n}{\ell}$ for all $\ell \leqs \frac{n}{6(d + 2)}$. Hence, for such $\ell$, $G_{\phi, \ell}$ contains a biclique of size $\binom{n}{\ell}$. \\

% $|L| = |R| = |\Sigma|^\ell \binom{n/2}{\ell} \geqs 3^\ell \binom{n/2}{\ell}$, which is at least $\binom{n}{\ell}$ for all $\ell \leqs n/4$. Hence, for such $\ell$, $G_{\phi, \ell}$ contains a biclique of size $\binom{n}{\ell}$. 

Finally, note that there are several ways to define constraints within $X_1$ and $X_2$ so that $\val(\phi)$ is bounded away from one. For instance, we can make each side a random 2-XOR formula, which results in $\val(\phi) \leqs 1/2 + O(1/d)$. Thus, if we start from a non-boolean CSP, the largest gap we can hope to get is only two.

%Consider any non-boolean 2CSP instance on variables $x_1, \dots, x_n$ such that there is no constraint between $\{x_1, \dots, x_{n/2}\}$ and $\{x_{n/2 + 1}, \dots, x_{n}\}$. Let $L$ and $R$ be the sets of all vertices corresponding to partial assignments to subsets of $\{x_1, \dots, x_{n/2}\}$ and $\{x_{n/2 + 1}, \dots, x_{n}\}$ respectively. Clearly, $(L, R)$ forms a biclique. Moreover, $|L| = |R| = |\Sigma|^\ell \binom{n/2}{\ell} \geqs 3^\ell \binom{n/2}{\ell}$, which is at least $\binom{n}{\ell}$ for all $\ell \leqs n/4$. Hence, for such $\ell$, $G_{\phi, \ell}$ contains a biclique of size $\binom{n}{\ell}$. Finally, we can add arbitrary constraints within $\{x_1, \dots, x_{n/2}\}$ and $\{x_{n/2 + 1}, \dots, x_{n}\}$ so that the value of the instance is bounded away from one. Thus, if we start from a non-boolean CSP, the largest gap we can hope to get is only two.

Note that the instance above is rather extreme as it consists of two disconnected components. Hence, it is still possible that, if the starting CSP has more specific properties (e.g. expanding constraint graph), then one can arrive at a gap of more than two.

\section{Proof of Lemma~\ref{lem:almost-eth-clauses}} \label{app:almost-eth-clauses}

To prove Lemma~\ref{lem:almost-eth-clauses}, we use the Sparsification Lemma of Impagliazzo~\etal~\cite[Corollary~1]{IPZ01}, which, for any constant $\sigma > 0$, reduces a $k$-SAT formula on $n$ variables in $\poly(n) \cdot 2^{\sigma n}$ time to the disjunction of at most $\poly(n) \cdot 2^{\sigma n}$ $k$-SAT formulas on $n$ variables, each with $O_{k, \sigma}(n)$ clauses. The same algorithm also applies to subconstant $\sigma > 0$, where for 3SAT Calabro~\etal~\cite[Section~6]{CIP06} showed that each output formula has $O(n (\sigma^{-1} \log(2/\sigma))^2) \leq O(n / \sigma^3)$ clauses (see also~\cite[Lemma~6.3]{Ope26ETH}). Specializing to 3SAT with $\sigma = 1/s$ for an integer $s \in \{1, \dots, n\}$ gives the following lemma:

\begin{lemma}[Sparsification Lemma~\cite{IPZ01,CIP06}] \label{lem:sparsification}
Given a 3SAT formula $\varphi$ on $n$ variables and an integer $s \in \{1, \dots, n\}$, there is a $\poly(n) \cdot 2^{n / s}$-time algorithm that outputs $t \leqs \poly(n) \cdot 2^{n / s}$ 3SAT formulas $\varphi_1, \dots, \varphi_t$ on $n$ variables such that
\begin{itemize}
\item $\varphi$ is satisfiable if and only if at least one $\varphi_i$ is satisfiable, and,
\item each $\varphi_i$ contains at most $O(n s^3)$ clauses.
\end{itemize}
\end{lemma}

\begin{proof}[Proof of Lemma~\ref{lem:almost-eth-clauses}]
Assume w.l.o.g. that $\gamma \in (0, 1)$ is rational, and suppose for the sake of contradiction that 3SAT with $m$ clauses can be decided in $O(2^{m^{1 - \gamma}})$ time. Given a 3SAT formula $\varphi$ on $n$ variables, we apply Lemma~\ref{lem:sparsification} with $s = \lfloor n^{\gamma / 4} \rfloor$ to obtain, in $2^{O(n^{1 - \gamma / 4})}$ time, $t \leqs 2^{O(n^{1 - \gamma / 4})}$ 3SAT formulas $\varphi_1, \dots, \varphi_t$, each with $m_i = O(n s^3) = O(n^{1 + 3\gamma / 4})$ clauses. Since $(1 + 3\gamma / 4)(1 - \gamma) < 1 - \gamma / 4$, deciding each $\varphi_i$ takes $O(2^{m_i^{1 - \gamma}}) = 2^{o(n^{1 - \gamma / 4})}$ time. Thus, $\varphi$ can be decided in $2^{O(n^{1 - \gamma / 4})}$ time, contradicting Almost-ETH.
\end{proof}

\end{document}